\documentclass[twocolumn]{autartMF}    

\usepackage{epsfig} 
\usepackage{mathptmx} 
\usepackage{times} 
\usepackage{pgfplots}

\usepackage{amssymb}
\usepackage{amsfonts}
\usepackage{amsmath}
\usepackage{latexsym}
\usepackage{graphicx}
\usepackage{color}
\usepackage{float}

\usepackage{epstopdf}
\usepackage{psfrag}
\usepackage{pstricks}
\usepackage{pstricks-add}
\usepackage{graphics} 
\usepackage{pgfplots}
\usetikzlibrary{positioning}

\usepackage{textcomp}

\usepackage[ruled]{algorithm}
\usepackage{algpseudocode}

\def\bO{\bar{\Omega}}

\def\bH{\bar{H}}

\def\bG{\bar{G}}
\def\bg{\bar{g}}
\def\bx{\bar{x}}
\def\ee{\hspace{-0cm} & \hspace{-0cm} }

\everymath{\displaystyle}

\newtheorem {theorem}{Theorem}
\newtheorem {lemma}{Lemma}

\newtheorem {assumption}{Assumption}
\newtheorem {remark}{Remark}

\def\R{\mathbb{R}}
\def\B{\mathbf{B}}
\def\N{\mathbb{N}}

\def\B{\mathcal{B}}

\def\co{\mathrm{co}}

\def\int{\text{int} }
\def\proj{\text{proj}}

\begin{document}

\begin{frontmatter}
\title{Computing control invariant sets in high dimension is easy} 

\author[Gipsa]{Mirko Fiacchini}\ead{mirko.fiacchini@gipsa-lab.fr},    
\author[Gipsa]{Mazen Alamir}\ead{mazen.alamir@gipsa-lab.fr},   

\address[Gipsa]{Univ. Grenoble Alpes, CNRS, Gipsa-lab, F-38000 Grenoble, 
France.} 


\begin{abstract}
In this paper we consider the problem of computing control invariant sets for 
linear controlled high-dimensional systems with constraints on the input and on 
the states. Set inclusions conditions for control invariance are presented that 
involve the N-step sets and are posed in form of linear programming problems. 
Such conditions allow to overcome the complexity limitation inherent to the set 
addition and vertices enumeration and can be applied also to high dimensional 
systems. The efficiency and scalability of the method are illustrated by 
computing approximations of the maximal control invariant set, based on the 
10-step operator, for a system whose state and input dimensions are 30 and 15, 
respectively.
\end{abstract}

\end{frontmatter}

\section{Introduction}

Invariance and contractivity of sets are central properties in modern control 
theory. Although the first important results on invariance date back to the 
beginning of the seventies \cite{BertsekasTAC72}, this topic gained 
considerable interest in the recent years, see in particular the works 
by Blanchini and coauthors \cite{BlanchiniTAC94,BlanchiniBook}, mainly due to 
its relation with constrained control and popular optimization-based control 
techniques as Model Predictive Control, see \cite{MayneAUT00}. 

Iterative procedures are given for the computation of control invariant sets 
that permit their practical implementation. Most of those procedures are 
substantially based on the one-step backward operator that associates to any set 
the states that can be steered into by an admissible input. Different algorithms 
based on the one-step operator exist for computing control invariants, that 
substantially differs from the initial set. For instance, if the algorithms are 
initialized with the state constraints set, 
\cite{BlanchiniTAC94,KerriganPHD01,RunggerTAC17}, the one-step operator 
generates a sequence of outer approximations of the maximal control invariant 
that converges to it under compactness assumptions, see \cite{BertsekasTAC72}. 
If, instead, the procedure is initialized with a control invariant set, a 
non-decreasing sequence of control invariant sets are obtained that converges 
from the inside to the maximal control invariant set, see the considerations on 
minimum-time ultimate boundedness problem in 
\cite{BlanchiniCDC92,BlanchiniBook}. A particular case, suggests to initialize 
the procedure with the set containing the origin only (which is a control 
invariant in the general framework), obtaining the sequence of $i$-step 
null-controllable sets, that are control invariant and converges to the maximal 
control invariant set, see \cite{GutTAC86,KeerthiTAC87,MayneAUT97,DarupCDC14}. 

Thus, although the abstract iterative procedures for obtaining control 
invariant sets apply also for nonlinear systems, 
\cite{FiacchiniAUT10,FiacchiniSCL12}, the practical computation of the 
one-step set, that is the basis for them, is often prohibitively complex for 
their application in high dimension even in the linear context. Some 
constructive approaches are based on Minkowski sum and projection procedure, as 
in \cite{KeerthiTAC87,BlanchiniCDC92,BlanchiniIJC95}, which are hardly 
applicable in high dimension due to their numerical complexity. Other methods 
are based on conditions involving the vertices of the sets under analysis,  
\cite{GutTAC86,LasserreAUT93,MayneAUT97,RakovicBaricTAC10}, but the vertices 
number may grow combinatorially with the space dimension and the vertices 
computation is hardly manageable in high dimension. The numerical complexity has 
also been addressed by considering linear feedback and ellipsoidal control 
invariant sets, see the monograph \cite{Boyd94}, or by fixing the polyhedral set 
complexity \cite{BlancoIJC10,AthanasopoulosIJC14,TahirTAC15}. 

In this paper we address the main problem related to the complexity of the 
N-step operator, for discrete-time deterministic controlled systems, with 
polyhedral constraints on the input and on the state. Considering polyhedral 
sets, such operator can be expressed in terms of Minkowski sum of polyhedra and 
then as an NP-complete problem \cite{Tiwary08Hardness}, hardly manageable in 
high dimension. An algorithm is presented for determining control invariant sets 
that is based on a set inclusion condition involving the N-step set of a 
polyhedron but does not require to explicitly compute the Minkowski sum nor to 
have the vertices representation of the sets. Such condition is posed as an LP 
feasibility problem, hence solvable even in high dimension. Examples that show 
the low conservatism and the high scalability of the approach are provided.

\paragraph*{Notations}
Denote with $\R_+$ the set of nonnegative real numbers. Given $n \in \N$, 
define $\N_n = \{x \in \N : 1 \leq x \leq n \}$. The $i$-th element of a 
finite set of matrices or vectors is denoted as $A_i$. Using the notation from 
\cite{RockafellarVariational}, given a mapping $M : \R^n \rightrightarrows 
\R^m$, its inverse mapping is denoted $M^{-1} : \R^n \rightrightarrows \R^m$. 
If $M$ is a single-valued linear mapping, we also denote, with slight abuse of 
notation, the related matrices $M \in \R^{n \times m}$ and, if $M$ is 
invertible, $M^{-1} \in \R^{m \times n}$. Given $a \in \R^n$ and $b \in \R^m$ 
we use the notation $(a, b) = [a^T \ b^T]^T \in \R^{n+m}$. The symbol $0$ 
denotes, besides the zero, also the matrices of appropriate dimensions whose 
entries are zeros and the origin of a vectorial space, its meaning being 
determined by the context. The symbol $\mathbf{1}$ denotes the vector of 
entries $1$ and $I$ the identity matrix, their dimension is determined by the 
context. The subset of $\R^n$ containing the origin only is $\{0\}$. The symbol 
$\oplus$ denotes the Minkowski set addition, i.e. given $C, D \subseteq \R^n$ 
then $C \oplus D = \{x + y \in \R^n: \ x \in C, \ y \in D\}$. To simplify the 
notation, the propositions involving the existential quantifier in the 
definition of sets are left implicit, e.g. $\{x \in A: f(x,y) \leq 0, \ y \in B 
\}$ means $\{x \in A: \ \exists y \in B \ \mathrm{ s.t. } f(x,y) \leq 0\}$. The 
unit box in $\R^n$ is denoted $\B^n$.

\section{Problem formulation and preliminary results}

The objective of this paper is to provide a constructive method to compute a 
control invariant set for controlled linear systems with constraints on the 
input and on the state. We would like to obtain a polytopic invariant set that 
could be computed through convex optimization problems. The main aim is to 
provide a method to obtain admissible control invariant sets for 
high-dimensional systems, thus no complex computational operations are supposed 
to be allowed.

The system is given by 
\begin{equation}\label{eq:system}
 x^+ = A x + B u
\end{equation}
where $x \in \R^n$ is the state and $u \in \R^m$ is the input, with constraints
\begin{equation}\label{eq:XU}
\begin{array}{l}
x \in X = \{y \in \R^n : \ F y\leq f\}, \ \ 
u \in U = \{v \in \R^m : \ G v \leq g\}.
\end{array}
\end{equation}

\begin{assumption}\label{ass:A}
The matrix $A$ is nonsingular. 
\end{assumption}

Assumption~\ref{ass:A}, not necessary but imposed here to easy the 
presentation, is not very restrictive. Recall for instance that every 
discretized linear system with no delay satisfies it. Anyway, the case of 
nonsingular $A$ is developed in \cite{FiacchiniArXiv17}.

Some basic properties and methods, well assessed in the literature, concerning 
control invariant sets are recalled hereafter. Consider the set $\Omega$ 
containing the origin, i.e. $0 \in \Omega$, and $Q_k(\Omega,U)$ defined as 
\begin{equation}\label{eq:Ok}
\begin{array}{l}
Q_k(\Omega,U) = \{x \in \R^n : \ A^k x + \sum_{i = 0}^{k-1} A^{k-1-i} B u_{k-i} 
\in \Omega, \\ 
\hspace{2cm} u_i \in U \ \forall i \in \N_k\}.
\end{array}
\end{equation}
The basic algorithm for obtaining a control invariant set consists in 
searching, given $\Omega$, for the minimal $N$ such that 
\begin{equation}\label{eq:invariance_cond1}
\Omega \subseteq \co \left(\bigcup_{k = 1}^{N} Q_k (\Omega,U)\right).
\end{equation}
As a matter of fact, all the $N$ for which (\ref{eq:invariance_cond1}) holds, 
lead to a control invariant set. Moreover, if (\ref{eq:invariance_cond1}) is 
satisfied, then it is satisfied for every $K \geq N$, leading to a 
non-decreasing sequence of nested control invariant sets.

Thus, the algorithm computes the preimages of $\Omega$ until the stop condition 
(\ref{eq:invariance_cond1}) holds and then all the states in 
\begin{equation}\label{eq:Omegainfty}
\bar{Q}_N(\Omega,U) = \co \left(\bigcup_{k = 1}^{N} Q_k (\Omega,U)\right)
\end{equation}
can be steered in $\Omega$, thus in $\bar{Q}_N(\Omega,U)$, in $N$ steps at 
most.

Given the initial set $\Omega$, a condition characterizing an invariant set, 
alternative to (\ref{eq:invariance_cond1}), is the following 
\begin{equation}\label{eq:invariance_condN}
\Omega \subseteq Q_N(\Omega,U),
\end{equation}
which is equivalent to the fact that every state in $Q_N(\Omega,U)$ can be 
steered in $\Omega$ in exactly $N$ steps. This means that 
(\ref{eq:invariance_condN}) implies, but is not equivalent to 
(\ref{eq:invariance_cond1}) and the resulting invariant set would be 
$\bar{Q}_N(\Omega,U)$ as in (\ref{eq:Omegainfty}). Condition 
(\ref{eq:invariance_condN}), which will be referred to as N-step condition in 
what follows, is just sufficient for (\ref{eq:invariance_cond1}) to hold but it 
does not require the computation of the convex hull of several sets at every 
iteration. 

Now, to obtain estimations of the maximal control invariant set contained 
in $X$, consider 
\begin{equation}\label{eq:OkX}
\begin{array}{l}
Q_k^x(\Omega,U,X) = \{x \in X : \ A^k x + \sum_{i = 0}^{k-1} A^{k-1-i} B 
u_{k-i} \in \Omega, \\
\hspace{1cm} A^j x + \sum_{i = 0}^{j-1} A^{j-1-i} B u_{j-i} \in X \ \ \forall j 
\in \N_k,\\
\hspace{1cm} u_i \in U \ \ \forall i \in \N_k\},
\end{array}
\end{equation}
that is the set of states $x \in X$ for which an admissible sequence of input 
of length $k$ exists driving the state in $\Omega$ in $k$ steps by maintaining 
the trajectory in $X$. The resulting control invariant set would then be given 
by 
\begin{equation}\label{eq:OmegainftyX}
\bar{Q}_N^x(\Omega,U,X) = \co \left(\bigcup_{k = 1}^{N} Q_k^x(\Omega,U,X) 
\right),
\end{equation}
provided that condition 
\begin{equation}\label{eq:invariance_condN_X}
\Omega \subseteq Q_N^x(\Omega,U,X)
\end{equation}
holds. Note that (\ref{eq:invariance_condN_X}) is just sufficient but, in 
general, less complex to be checked than $\Omega \subseteq 
\bar{Q}_N^x(\Omega,U,X)$.

\begin{remark}
The value of $N$ for which invariance condition (\ref{eq:invariance_condN}) and 
(\ref{eq:invariance_condN_X}) hold depends on the choice of $\Omega$. Clearly, 
if $\Omega$ is a control invariant set, then the conditions hold for all $N 
\geq 1$. Moreover, for every $\Omega$ there exists $\alpha > 0$ and $N \geq 1$ 
such that $\alpha \Omega$ satisfies (\ref{eq:invariance_condN}) or 
(\ref{eq:invariance_condN_X}), under mild stabilizability conditions.
\end{remark}

\subsection{Algorithm}

The main issue which impedes the application of the algorithm in high dimension 
is the fact that computing the Minkowski set addition is a complex operation, 
as it is an NP-complete problem, see \cite{Jones04,Tiwary08Hardness}. Moreover 
the addition leads to sets whose representation complexity increases. 
Considering, in fact, two polytopic sets $\Omega$ and $\Delta$, their sum has in 
general more facets and vertices than $\Omega$ and $\Delta$. Thus, the algorithm 
given above requires the computation of the Minkowski sum, hardly manageable in 
high dimension, and generates polytopes with an increasing number of facets and 
vertices. Another source of complexity is the convex hull in 
(\ref{eq:invariance_cond1}), (\ref{eq:Omegainfty}) or (\ref{eq:OmegainftyX}), as 
the explicit computation of the convex hull is a non-convex operation whose 
complexity grows exponentially with the dimension, see 
\cite{DeBerg2000computational}. Furthermore, also the vertices representation of 
the sets is a potential limitation for high dimensional systems, since the 
number of vertices may grow combinatorially with the dimension. Finally, 
approaches are provided, for instance in \cite{KeerthiTAC87,BlanchiniIJC95}, 
that require the computation of the projection of polytopes, operation whose 
complexity is equivalent to the one of Minkowski sum. As can be seen from the 
comparison, provided in \cite{Jones04}, between different projection 
algorithms, polytope projections are not suitable when projecting on high 
dimensions. This can be also heuristically checked by computing the projection 
over an $n$-dimensional subspace of a randomly generated $2n$-dimensional 
polytope. Using the MPT toolbox \cite{mpt}, for instance, we needed more than 
$40$ seconds to project a polytope from $\R^{10}$ into a $5$-dimensional 
subspace, more than $15$ minutes to project from $\R^{12}$ to $\R^{6}$.

The main objective of this paper is to design a method for testing conditions 
(\ref{eq:invariance_condN}) and (\ref{eq:invariance_condN_X}) and for having 
a, potentially implicit, representation of sets (\ref{eq:Omegainfty}) and 
(\ref{eq:OmegainftyX}) by means of convex optimization problems, then 
applicable also to relatively high dimensional systems, to obtain control 
invariant sets, avoiding the vertices representation of the sets and Minkowski 
sum or polytope projections computation.

\section{N-step condition for control invariance}

As noticed above, a first main issue is related to checking whether the sum of 
several polytopes contains a polytope, see the $N$-step stop condition 
(\ref{eq:invariance_condN}) and (\ref{eq:invariance_condN_X}). 

Consider first the N-step condition (\ref{eq:invariance_condN}), characterized 
by the Minkowski sum of several sets. The explicit definition of the Minkowski 
sum of sets could be avoided by employing its implicit representation. Indeed, 
given two polyhedral sets $\Gamma = \{x \in \R^m : H x \in h\}$ and $\Delta = 
\{y \in \R^p: G y \leq g\}$ and $P \in \R^{n \times m}$ and $T \in \R^{n \times 
p}$ we have that $P \Gamma \oplus T \Delta = \{x \in \R^n: x = P y + T z, \ Hy 
\leq h, \ Gz\leq g \}$. Thus, the explicit hyperplane or vertex representation 
of the sum can be replaced by the implicit one, given by the projection of a 
polyhedron in higher dimension. On the other hand, one might wonder if the stop 
condition $\Omega \subseteq Q_{N}(\Omega,U)$ could be checked without the 
explicit representation of $Q_{N}(\Omega,U)$. 

The first hint to do is that the inclusion condition is testable through a set 
of LP problems provided the vertices of $\Omega$ are available. Such an 
assumption is not very restrictive, since $\Omega$ is a design parameter that 
could be determined such that both the hyperplane and vertices representation 
should be available, a box for instance. Nevertheless, and since we are aiming 
at invariant sets for high dimensional systems, the use of vertices should be 
avoided if possible. Consider for instance, in fact, a system with $n = 20$. 
The unit box in $\R^{20}$ is characterized by 40 hyperplanes, but it has 
$2^{20} \simeq 10^6$ vertices. Then checking if it is contained in a set could 
require to solve more than a million of LP problems.

We consider then the possibility of testing whether a polyhedron is included in 
the sum of polyhedra by employing only their hyperplane representations and 
without the explicit representation of the sum of sets. The following result, 
based on the Farkas lemma and widely used on set theory and invariant methods 
for control, is useful for this purpose.

\begin{lemma}\label{lem:Farkas}
Two polyhedral sets $\Gamma = \{x \in \R^n : H x \leq h\}$, with $H \in \R^{p 
\times n}$, and $\Delta = \{x \in \R^n: G x \leq g\}$, with $G \in \R^{q \times 
n}$, satisfy $\Gamma \subseteq \Delta$ if and only if there exists a 
non-negative matrix $T \in \R^{q \times p}$ such that $T H = G$ and $T h \leq 
g$.
\end{lemma}
Consider now the stop condition (\ref{eq:invariance_condN}), which is suitable 
for applying the Lemma \ref{lem:Farkas}, as illustrated below. 

The main issue for applying Lemma \ref{lem:Farkas} is the fact that obtaining 
the explicit hyperplane representation of the set at right-hand side of 
(\ref{eq:invariance_condN}) is numerically hardly affordable, mainly in 
relatively high dimension. In fact, given two polyhedra $\Gamma \subseteq \R^m$ 
and $\Delta \subseteq \R^p$, to determine $L$ and $l$ such that $P \Gamma 
\oplus Q \Delta = \{x \in \R^n : \ L x \leq l\}$ is an NP-complete problem, see 
\cite{Tiwary08Hardness}. Nevertheless, a sufficient condition in form of LP 
feasibility problem is given below.

\begin{theorem}\label{th:Ninclusion}
Consider $\Omega = \{x \in \R^n: Hx \leq h\}$ and $U$ as in (\ref{eq:XU}), 
with $H \in \R^{n_h \times n}$ and $G \in \R^{n_g \times m}$, and suppose that 
$0 \in \Omega$ and $0 \in U$. Then the set $\bar{Q}_N(\Omega,U)$ as in 
(\ref{eq:Omegainfty}) is a control invariant set if there exist $T \in 
\R^{n_{\bar{g}} \times n_h}$ and $M \in \R^{\bar{n} \times \bar{n}}$, with 
$n_{\bg} = n_h+Nn_g$ and $\bar{n} = 
n+Nm$, such that 
\begin{equation}\label{eq:FinalCondLP}
\left\{\begin{array}{l}
T \bH = \bG M\\
T h \leq \bg\\
\left[\begin{array}{ccccc} I \ee 0 \ee 0 \ee \ldots \ee 0 \end{array}\right] =  
\left[\begin{array}{ccccc} I \ee 0 \ee 0 \ee \ldots \ee 0\end{array}\right] M
\end{array}\right.
\end{equation}
hold with 
\begin{equation}\label{eq:bGinv}
\begin{array}{l}
\bG = \left[\begin{array}{ccccc}
H A^N \ee H B \ee HAB \ee \ldots \ee HA^{N-1}B\\
0 \ee G \ee 0 \ee \ldots \ee 0\\
0 \ee 0 \ee G \ee \ldots \ee 0\\
 \ldots \ee \ldots  \ee \ldots  \ee \ldots \ee 
\ldots \\
0 \ee 0 \ee 0 \ee \ldots \ee G\\
\end{array}\right] \hspace{-0.1cm}, \ \
\bg = 
\left[\begin{array}{c}
h\\
g\\
g\\
\ldots\\
g
\end{array}\right]\\ 
\bH = \left[\begin{array}{cccccc}
H \ee 0 \ee 0 \ee \ldots \ee 0
\end{array}\right] 
\end{array}
\end{equation}
where $\bG \in \R^{n_{\bg} \times \bar{n}}$, $\bg \in \R^{n_{\bg}}$, and $\bH 
\in \R^{n_{h} \times \bar{n}}$.
\end{theorem}

\begin{proof}
Consider condition (\ref{eq:invariance_condN}), sufficient for 
$\bar{Q}_N(\Omega,U)$ to be a control invariant set. The right-hand side term 
of (\ref{eq:invariance_condN}) is given by 
\begin{equation*}
\begin{array}{l}
\displaystyle Q_N(\Omega,U) = \{x \in \R^n : \, H A^N x + H B u_1  + H A^{1} 
Bu_2 + \ldots \\
\hspace{2cm} + H A^{N-1} B u_N  \leq h, \ G u_i \leq g, \ \forall i \in \N_N \}
\end{array}
\end{equation*}
and then it is the projection on $\R^n$ of the set 
\begin{equation*}\label{eq:OplusinvZ}
\begin{array}{l}
\bar{\Omega}_{N} = \{(x, u_{1}, u_{2}, \ldots, u_{N}) \in \R^{\bar{n}}: \ H A^N 
x + H B u_1 + H A B  u_2 \ldots \\
+ H A^{N-1}Bu_N \leq h, \ G u_i \leq g \ \forall i 
\in \N_N\} = \{\bx \in \R^{\bar{n}}: \ \bG \bx \leq \bg\},
\end{array}
\end{equation*}
with $\bx = (x, u_{1}, u_{2}, \ldots, u_{N}) \in \R^{\bar{n}}$ and $\bG$ and 
$\bg$ as in (\ref{eq:bGinv}). The set $\Omega$ is the projection on $\R^n$ of 
the set 
\begin{equation*}
\begin{array}{cl}
\bO \hspace{-0.1cm} & = \{(x, u_{1}, u_{2}, \ldots, u_{N}) \in \R^{\bar{n}}: \ 
H x \leq h\} \\
& = \{\bx \in \R^{\bar{n}}: \ \bH \bx \leq h\} \subseteq \R^{\bar{n}}
\end{array}
\end{equation*}
with $\bH$ as in (\ref{eq:bGinv}). Thus, condition (\ref{eq:invariance_condN}) 
is equivalent to
\begin{equation}\label{eq:projCNS}
\proj_x \bO \subseteq \proj_x \bO_N,
\end{equation}
since $\Omega = \proj_x \bO $ and $Q_N = \proj_x \bO_N$.
Thus, to prove that $\Omega \subseteq Q_N(\Omega,U)$ is equivalent to check 
whether the projection of $\bar{\Omega}_{N}$ on $\R^n$ contains the 
projection of $\bO$. Unfortunately, condition (\ref{eq:projCNS}) is not 
suitable for using Lemma \ref{lem:Farkas} and then we search for a sufficient 
condition for (\ref{eq:projCNS}) to hold such that the lemma can be applied 
directly. 

Consider any linear single-valued mapping $M : \R^{\bar{n}} \rightrightarrows 
\R^{\bar{n}}$, characterized by a, possibly non-invertible, matrix $M \in 
\R^{\bar{n} \times \bar{n}}$, such that the value of $x$ through $M$ is 
preserved, i.e. $\proj_x M((x, u_1, \ldots, u_N)) = x$ for all $(x, u_1, 
\ldots, u_N) \in \R^{\bar{n}}$. Clearly, the value of $x$ is preserved also 
through the inverse mapping of $M$, that is $\proj_x M^{-1}((x, u_1, \ldots, 
u_N)) = x$ for all $(x, u_1, \ldots, u_N) \in \R^{\bar{n}}$. This means that 
$\proj_x \bar{\Omega}_{N} = \proj_x M^{-1} \bar{\Omega}_{N}$ and 
then (\ref{eq:projCNS}) is equivalent to 
\begin{equation}\label{eq:projCNS2}
\proj_x \bar{\Omega} \subseteq \proj_x M^{-1} \bar{\Omega}_{N}.
\end{equation}
Then, the existence of $M$ preserving the $x$ and such that 
\begin{equation}\label{eq:projCNS3}
\bar{\Omega} \subseteq M^{-1} \bar{\Omega}_{N}
\end{equation}
holds, is a sufficient condition for (\ref{eq:projCNS2}), and thus also for 
(\ref{eq:projCNS}), to be satisfied. Notice that necessity of 
(\ref{eq:projCNS3}) for (\ref{eq:projCNS2}) is not straightforward, since 
$\proj_x \Gamma \subseteq \proj_x \Delta$ does not imply $\Gamma 
\subseteq \Delta$, in general. 

The condition on the matrix $M$ such that $\proj_x M((x, u_1, \ldots, u_N)) = 
x$ for all $(x, u_1, \ldots, u_N) \in \R^{\bar{n}}$ is 
\begin{equation}\label{eq:N}
\left[\begin{array}{cccc} I \ee 0 \ee \ldots \ee 0 \end{array}\right] =  
\left[\begin{array}{cccc} I \ee 0 \ee \ldots \ee 0 \end{array}\right] M
\end{equation}
and then, from Lemma \ref{lem:Farkas}, it follows that conditions 
(\ref{eq:projCNS3}) and (\ref{eq:N}) are equivalent to the existence of $T \in 
\R^{n_{\bar{g}} \times n_{\bar{h}}}$ and $M \in \R^{(\bar{n}) \times 
(\bar{n})}$ 
satisfying (\ref{eq:FinalCondLP}). Then (\ref{eq:FinalCondLP}) is a sufficient 
condition for $\Omega \subseteq Q_N(\Omega,U)$.
\end{proof}

The result given above can be directly extended to the problem in presence of 
constraints on the state.

\begin{theorem}\label{th:NinclusionX}
Consider $\Omega = \{x \in \R^n: Hx \leq h\}$ and $X$ and $U$ as in 
(\ref{eq:XU}), with $H \in \R^{n_h \times n}$, $F \in \R^{n_f \times n}$ $G \in 
\R^{n_g \times m}$, and suppose that $0 \in \Omega$, $0 \in X$ and $0 \in U$. 
Then the set $\bar{Q}_N^x(\Omega,U,X)$ as in (\ref{eq:OmegainftyX}) is a 
control invariant set contained in $X$ if there exist $T \in \R^{n_{\bar{g}} 
\times n_h}$ and $M \in \R^{\bar{n} \times \bar{n}}$, with $n_{\bg} = n_h + N 
n_g + N n_f$ and $\bar{n} = n+Nm$, such that (\ref{eq:FinalCondLP}) holds with 
\begin{equation}\label{eq:bGinvX}
\begin{array}{l}
\bG = \left[\begin{array}{ccccc}
H A^N \ee H B \ee HAB \ee \ldots \ee HA^{N-1}B\\
0 \ee G \ee 0 \ee \ldots \ee 0\\
0 \ee 0 \ee G \ee \ldots \ee 0\\
 \ldots \ee \ldots  \ee \ldots  \ee \ldots \ee \ldots \\
0 \ee 0 \ee 0 \ee \ldots \ee G\\
\hline
F A^N \ee FB \ee FAB \ee \ldots \ee FA^{N-1}B\\
F A^{N-1} \ee 0 \ee FB \ee \ldots \ee FA^{N-2}B\\
 \ldots \ee \ldots  \ee \ldots  \ee \ldots \ee \ldots \\
F A \ee 0 \ee 0 \ee \ldots \ee FB\\
F \ee 0 \ee 0 \ee \ldots \ee 0\\
\end{array}\right] \hspace{-0.1cm}, \ \
\bg = 
\left[\begin{array}{c}
h\\
g\\
g\\
\ldots\\
g\\
\hline
f\\
f\\
\ldots\\
f\\
f
\end{array}\right]\\ 
\bH = \left[\begin{array}{cccccc}
H \ee 0 \ee 0 \ee \ldots \ee 0
\end{array}\right] 
\end{array}
\end{equation}
where $\bG \in \R^{n_{\bg} \times \bar{n}}$, $\bg \in \R^{n_{\bg}}$, and $\bH 
\in \R^{n_{h} \times \bar{n}}$.
\end{theorem}

\begin{proof}
Condition (\ref{eq:FinalCondLP}) with (\ref{eq:bGinvX}) can be proved to imply 
the constrained invariant condition (\ref{eq:invariance_condN_X}) by reasonings 
analogous to those of Theorem~\ref{th:Ninclusion}.
\end{proof}

Given the sets $\Omega,U$ and $X$, to obtain the greatest multiple of 
$\Omega$ such that (\ref{eq:invariance_condN_X}) holds, that is the greatest 
$\alpha \in \R$ such that 
\begin{equation}\label{eq:invariance_cond2}
\alpha \Omega \subseteq Q_N^{x}(\alpha \Omega,U,X),
\end{equation}
is equivalent to compute the smallest nonnegative $\beta$, with $\beta = 
\alpha^{-1}$, such that 
\begin{equation*}\label{eq:invariance_cond3Z}
\Omega \subseteq Q_N^x(\Omega,\beta U, \beta X).
\end{equation*}
This consists in replacing $g$ with $\beta g$ in (\ref{eq:bGinvX}) and leads to 
the following LP problem in $T$, $M$ and $\beta$ 
\begin{equation}\label{eq:FinalCondLPbeta}
\begin{array}{l}
\hspace{-0.5cm} \alpha^{-1} = \beta = \min_{\beta \in \R_+} \gamma\\
\hspace{1.3cm} \mathrm{s.t. } \ \ T \bH = \bG M\\
\hspace{2.0cm} T h \leq \gamma \hat{g} + \tilde{g}\\
\hspace{2.0cm} \left[\begin{array}{ccccc} I \ee  0 
\ee 0 \ee  \ldots 
\ee 0 \end{array}\right] =  
\left[\begin{array}{ccccc} I \ee  0 
\ee 0 \ee  \ldots 
\ee 0\end{array}\right] M
\end{array}
\end{equation}
with $\hat{g} = (0, \, g, \, g, \,  \ldots, \, g, \, f, \,  \ldots, \, f, \, 
f)$ and $\tilde{g} = (h, \, 0, \, 0, \, \ldots, \, 0)$.
\begin{remark}
 Clearly, if $\Omega$ is a control invariant set, then the greatest $\alpha$ 
satisfying (\ref{eq:invariance_cond2}) is not smaller than $1$.
\end{remark}

Note that directly maximizing $\alpha$ would yield to replace $h$ by $\alpha h$ 
in (\ref{eq:FinalCondLP}) and (\ref{eq:bGinvX}) and then to a nonlinear 
optimization problem. Analogous computational considerations hold for the case 
of absence of state constraints, as in Theorem~\ref{th:Ninclusion}, which is a 
particular case of Theorem \ref{th:NinclusionX} with $X = \R^n$.

\section{State inclusion test}

In the previous section, a condition for (\ref{eq:invariance_condN_X}) to hold 
is given that does not require the computation of the preimage sets 
$Q_N^x(\Omega,U,X)$, then avoiding the computation of Minkowski addition, see 
Theorem \ref{th:NinclusionX}. Once $\alpha \Omega$ is computed by solving 
(\ref{eq:FinalCondLPbeta}), one possible choice to obtain a control invariant 
set is given by 
\begin{equation}\label{eq:OmegaNx}
\bar{\Omega}^x = \co\left(\bigcup_{k = 1}^N \Omega_k^x \right) \ \ \mathrm{ 
with } \ \ \Omega_k^x = Q_k^x(\alpha \Omega, U, X).
\end{equation}
To have an explicit representation of $\bar{\Omega}^x$ requires to compute 
the convex hull of the union of several sets, each one given by the Minkowski 
sum of sets, but the convex hull operation is numerically demanding. Hereafter 
we provide a convex condition to check if a given $x \in \R^n$ belongs to the 
invariant set $\bar{\Omega}^x$ without computing it explicitly.

The theorem below provides a representation of the control invariant 
$\bar{\Omega}^x$ in terms of linear equalities and inequalities.

\begin{theorem}\label{th:theoremOinf}
Let Assumption \ref{ass:A} hold. Consider $\Omega = \{x \in \R^n: Hx \leq h\}$ 
bounded, $X$ and $U$ as in (\ref{eq:XU}), with $H \in \R^{n_h \times n}$, $F \in 
\R^{n_f \times n}$, $G \in \R^{n_g \times m}$, and suppose that $0 \in \Omega$, 
$0 \in X$, $0 \in U$ and $U$ is bounded. Given $\alpha$ solution of 
(\ref{eq:FinalCondLPbeta}) then the set $\bar{\Omega}^x$ defined by 
(\ref{eq:OmegaNx}) can be written as follows
\begin{equation}\label{eq:Inv_for_Anonsin}
\begin{array}{l}
\bar{\Omega}^x = \{ x \in \R^n: \ x = \sum_{k = 1}^N z_k; \\ 
\hspace{0.5cm}H A^k z_k + \sum_{i = 0}^{k-1} H A^{k-1-i} B v_{k,k-i} \leq 
\alpha \lambda_k h, \ \forall k \in \N_N;\\
\hspace{0.5cm} F A^j z_k + \sum_{i = 0}^{j-1} F A^{j-1-i} B v_{k,j-i} 
\leq \lambda_k f, \ \forall j \in \N_k, \ \forall k \in \N_N;\\
\hspace{0.5cm} F z_k \leq \lambda_k f, \ \forall k \in \N_N; \ \ \ G v_{k,i} 
\leq \lambda_k g \ \forall i \in \N_k, \ \forall k \in 
\N_N;\\
\hspace{0.5cm}\lambda \geq 0, \ \sum_{k =1}^N \lambda_k = 1\}.
\end{array}
\end{equation}
\end{theorem}

\begin{proof}
Given an arbitrary collection of non-empty convex sets $\Gamma_i \subseteq 
\R^n$ 
with $I \in \N$ and $i \in \N_I$, note first that 
\begin{equation*}
\begin{array}{l}
\co\left(\bigcup_{i \in \N_I} \Gamma_i \right) = \hspace{-0.2cm} 
\bigcup_{\substack{\lambda \geq 0 \\ 
\mathbf{1}^T \lambda = 1}} \hspace{-0.2cm} \left( \bigoplus_{i \in \N_I} 
\lambda_i \Gamma_i \right)
\end{array}
\end{equation*}
see Chapter 3 in \cite{Rockafellar}. Provided condition 
(\ref{eq:invariance_cond2}) is satisfied and from Lemma~\ref{lem:lambdaK} in 
Appendix \ref{app:A}, the control invariant set is given by 
\begin{equation}
\begin{array}{l}
 \bar{\Omega}^x \displaystyle = \co \left(\bigcup_{k = 1}^{N} \Omega_k^x 
\right) = \bigcup_{\substack{\lambda \geq 0 \\ \mathbf{1}^T \lambda = 1}} 
\left( \bigoplus_{k = 1}^N \lambda_k \Omega_k^x \right) \\
= \bigcup_{K \subseteq \N_N} \bigcup_{\substack{\lambda \in \Lambda(K)  
\\ \mathbf{1}^T \lambda = 1}} \Bigg( \bigoplus_{k \in K} \lambda_k \Omega_k^x 
\Bigg) = \{ x \in \R^n: \ x = \sum_{k \in K} \lambda_k y_k; \\ y_k \in 
\Omega_k^x, \ \lambda_k > 0, \ \ \forall k \in K; \ \sum_{k \in K} 
\lambda_k = 1, \ \forall K \subseteq \N_N\},
\end{array}
\end{equation}
since $0 \cdot \Omega_k^x = \{0\}$ for all $k \in \N_n$. Then, from $\lambda_k > 
0$ for every $k \in K$ and defining $z_k = \lambda_k y_k$ for all $k \in K$, it 
follows 
\begin{equation}
\begin{array}{l}
 \bar{\Omega}^x = \{ x \in \R^n: \ x = \sum_{k \in K} z_k; \ z_k/\lambda_k \in 
\Omega_k^x, \ \lambda_k > 0, \ \ \forall k \in K; \\ 
\hspace{0.25cm} \sum_{k \in K} \lambda_k = 1, \ \forall K \subseteq \N_N\} = \{ 
x \in \R^n: \ x = \sum_{k \in K} z_k; \\ 
\hspace{0.25cm}H A^k z_k/\lambda_k + \sum_{i = 0}^{k-1} H A^{k-1-i} B u_{k,k-i} 
\leq \alpha h, \ \forall k \in K;\\
\hspace{0.25cm} F A^j z_k/\lambda_k + \sum_{i = 0}^{j-1} F A^{j-1-i} B 
u_{k,j-i} \leq f, \ \forall j \in \N_k, \ \forall k \in K;\\
\hspace{0.25cm} F z_k/\lambda_k  \leq f, \ \forall k \in K; \ \ \ G 
u_{k,i} \leq g \ \forall i \in \N_k, \ \forall k \in K;\\
\hspace{0.25cm}\lambda_k > 0, \ \ \forall k \in K; \ \ \sum_{k \in K} \lambda_k 
= 1, \ \forall K \subseteq \N_N\}\\ 
\end{array}
\end{equation}
from (\ref{eq:OkX}). By introducing $v_{k,i} = \lambda_k u_{k,i}$ for all $i 
\in \N_k$, $k \in K$ and $K \subseteq \N_N$, then 
\begin{equation}\label{eq:Omegaconvex}
\begin{array}{l}
 \bar{\Omega}^x  = \{ x \in \R^n: \ x = \sum_{k \in K} z_k; \\ 
\hspace{0.25cm}H A^k z_k + \sum_{i = 0}^{k-1} H A^{k-1-i} B v_{k,k-i} \leq 
\alpha \lambda_k h, \ \forall k \in K;\\
\hspace{0.25cm} F A^j z_k + \sum_{i = 0}^{j-1} F A^{j-1-i} B v_{k,j-i} 
\leq \lambda_k f, \ \forall j \in \N_k, \ \forall k \in K;\\
\hspace{0.25cm} F z_k \leq \lambda_k f, \ \forall k \in K; \ \ \ G 
v_{k,i} \leq \lambda_k g \ \forall i \in \N_k, \ \forall k \in 
K;\\
\hspace{0.25cm}\lambda_k > 0, \ \ \forall k \in K; \ \ \ 
\sum_{k \in K} \lambda_k = 1, \ \forall K \subseteq \N_N\}.
\end{array}
\end{equation}
If $\lambda_k = 0$, as for all $k \notin K$ and every $K \subseteq \N_N$, then 
$G v_{k,i} \leq \lambda_k g$ implies $v_{k,i} = 0$ for all $i \in \N_k$, from 
the boundedness of $U$. Thus $\lambda_k = 0$ implies also $z_k = 0$, from the 
boundedness of $\Omega$ and Assumption~\ref{ass:A}. Hence the expression 
(\ref{eq:Inv_for_Anonsin}) can be recovered by posing $\lambda_k = 0, v_{k,i} = 
0$ and $z_k = 0$ in (\ref{eq:Omegaconvex}) for all $k \notin K$ and every $K 
\subseteq \N_N$.
\end{proof}

Theorem~\ref{th:theoremOinf} implies that checking if $x \in 
\bar{\Omega}^x$ resorts to solve an LP feasibility problem in the 
variables $x, z_k, v_{k,i}, \lambda_k$ for all $i \in \N_k$ and $k \in \N_N$, 
then in a space of dimension $n + Nn + 0.5 N(N+1) m + N$. Such a representation 
is particularly suitable to be used in optimization-based control, as model 
predictive control for instance, since it reduces to enforcing the linear 
constraints characterizing $\bar{\Omega}^x$.

\section{Numerical examples}

The different results presented in this paper are illustrated through numerical 
examples. The optimization problems are solved using YALMIP interface 
\cite{Lofberg2004} and Mosek optimizer \cite{mosek} on an 
Intel\textregistered{} Core\texttrademark{} i7-6600U CPU @ 2.60GHz \texttimes{} 
4 processor laptop with 16GB of RAM.

\subsection{Example 1}

Here we compare the computational burden required to check the invariant 
condition (\ref{eq:invariance_cond2}) by solving 
(\ref{eq:bGinvX})-(\ref{eq:FinalCondLPbeta}) with an alternative approach based 
on known properties of computational geometry. First we describe this approach. 
Suppose that both the hyperplanes and the vertices representation of $\Omega$ 
are available. This assumption, not needed for our method that only requires 
the H-representation, could be reasonably posed since $\Omega$ can be arbitrary 
chosen, and then fixed to be the unitary box, i.e. $\Omega = \B^n$. Then, the 
$2^n$ vertices can be easily obtained, which is not the case for general 
polytopes. Thus, the exact maximal $\alpha$ such that 
(\ref{eq:invariance_cond2}) is satisfied is given by
\begin{equation}\label{eq:alternativemethod}
\begin{array}{l}
\alpha^* = \max_{\substack{\alpha, \ u_{i,j}}} \alpha\\
\hspace{0.5cm} \mathrm{s.t. } \ \ \bG \cdot (\alpha v_j, \, u_{1,j}, \, \ldots, 
\, u_{N,j}, ) \leq \alpha \tilde{g} + \hat{g}, \ \ \forall j \in \N_{2^n}
\end{array}
\end{equation}
where $v_j$ is the $j$-th vertex, with $j \in \N_{2^n}$, and $\bG, \tilde{g}$ 
and $\hat{g}$ are defined in and below Theorem~\ref{th:NinclusionX}. The 
constraints in the LP problem (\ref{eq:alternativemethod}) impose that every 
vertex of $\alpha \Omega$ is contained in $Q_N^x(\alpha \Omega, U, X)$ and 
their 
number is equal to the number of vertices of $\Omega$, hence 
exponentially growing with the system dimension. 

The exact maximal $\alpha^*$ solution of (\ref{eq:alternativemethod}) and the 
$\alpha$ obtained by solving (\ref{eq:FinalCondLPbeta}) with (\ref{eq:bGinvX}) 
are computed for randomly generated controllable systems with real eigenvalues 
with increasing $n$ and $m = \lceil n/2 \rceil$. The sets are $\Omega = \B^n, 
U = 10 \B^m$ and $X = 100 \B^n$ and $N = 2$. The computation times are given in 
Figure~\ref{fig:Ex1_1} in function of the state dimension $n$.

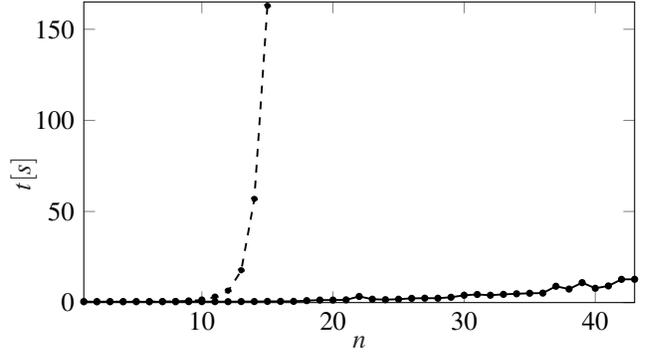
\begin{figure}[H]
\begin{center}
    \newlength\fheight
    \newlength\fwidth
    \setlength\fheight{4cm}
    \setlength\fwidth{7.7cm}
%
%
\begin{tikzpicture}

\begin{axis}[%
width=0.951\fwidth,
height=\fheight,
at={(0\fwidth,0\fheight)},
scale only axis,
xmin=1,
xmax=43,
xlabel style={at={(axis description cs:0.5,0.05)}, anchor=north, 
font=\color{white!15!black}},
xlabel={$n$},
ymin=0,
ymax=165,
ylabel style={at={(axis description cs:0.1,0.43)}, anchor=south, 
font=\color{white!15!black}},
ylabel={$t [s]$},
axis background/.style={fill=white}
]
\addplot [mark=*, mark size=1pt, color=black, line width=0.25mm, forget plot]
  table[row sep=crcr]{%
1	0.462802\\
2	0.426686\\
3	0.463526\\
4	0.429506\\
5	0.469741\\
6	0.438571\\
7	0.468756\\
8	0.466463\\
9	0.528841\\
10	0.524614999999997\\
11	0.533923999999999\\
12	0.503406000000005\\
13	0.576425999999998\\
14	0.574491000000002\\
15	0.656109999999998\\
16	0.665360999999997\\
17	0.680012999999995\\
18	1.078272\\
19	1.38209\\
20	1.372108\\
21	1.498352\\
22	3.368542\\
23	1.862068\\
24	1.60081899999999\\
25	1.855628\\
26	2.317049\\
27	2.416777\\
28	2.391974\\
29	2.918479\\
30	4.023089\\
31	4.456763\\
32	4.0272\\
33	4.520795\\
34	4.826656\\
35	5.11983\\
36	5.203598\\
37	8.955385\\
38	7.418312\\
39	10.941483\\
40	7.856905\\
41	9.136984\\
42	12.74389\\
43	12.785738\\
};
\addplot [mark=*, mark size=1pt, color=black, dashed, line width=0.25mm, forget 
plot]
  table[row sep=crcr]{%
1	0.489182999999997\\
2	0.484047999999994\\
3	0.484685999999996\\
4	0.489401000000001\\
5	0.492873000000003\\
6	0.504061999999998\\
7	0.535950999999997\\
8	0.620144\\
9	0.904401\\
10	1.420166\\
11	3.00205\\
12	6.280552\\
13	17.454025\\
14	56.600211\\
15	162.609212\\
};
\end{axis}
\end{tikzpicture}%
    \vspace{-0.35cm}
    \caption{Computation times in seconds to solve 
(\ref{eq:alternativemethod}), 
in dashed line, and to solve (\ref{eq:FinalCondLPbeta}) with (\ref{eq:bGinvX}),
in solid line, in function of the state dimension $n$.}
\label{fig:Ex1_1}
\end{center}
\end{figure}

Note that the proposed method permits to check the invariance condition up to a 
40 dimensional system with 20 inputs in less than 12 seconds, while the 
alternative approach needs more than $160 s$ for $n = 15$. 

In Figure \ref{fig:Ex1_2} we report the values of $\alpha$ obtained by solving 
(\ref{eq:FinalCondLPbeta}) for $1000$ randomly generated systems with 
$n$ between $1$ and $12$ and also the normalized mismatch with respect to 
$\alpha^*$ given by (\ref{eq:alternativemethod}), i.e. $|\alpha - 
\alpha^*|/\alpha^*$, in logarithmic scale. It can be noticed that the 
approximation error is several order of magnitude smaller than the values of 
$\alpha$, in most of the cases included between $10^{-4}$ and $10^{-12}$, which 
might be due to the numerical precision rather than to real inaccuracy.

\begin{figure}[H]
\begin{center}
    \setlength\fheight{3cm}
    \setlength\fwidth{8cm}
%
%
\begin{tikzpicture}

\begin{axis}[%
width=0.951\fwidth,
height=\fheight,
at={(0\fwidth,0\fheight)},
scale only axis,
xmin=-14,
xmax=2,
xlabel style={at={(axis description cs:0.5,0.05)}, anchor=north, 
font=\color{white!15!black}},
xlabel={$\hspace*{1cm} \log(|\alpha - \alpha^*|/\alpha^*) \hspace{2cm} 
\log(\alpha)$},
ymin=0,
ymax=200,
axis background/.style={fill=white}
]
\addplot[ybar interval, fill=darkgray, draw=black, area legend] table[row 
sep=crcr] {%
x	y\\
-12.4041168130844	1\\
-12.1160344768228	1\\
-11.8279521405611	1\\
-11.5398698042994	5\\
-11.2517874680377	4\\
-10.963705131776	14\\
-10.6756227955143	15\\
-10.3875404592526	16\\
-10.0994581229909	29\\
-9.81137578672924	29\\
-9.52329345046755	42\\
-9.23521111420586	37\\
-8.94712877794418	64\\
-8.65904644168249	73\\
-8.3709641054208	86\\
-8.08288176915911	97\\
-7.79479943289742	94\\
-7.50671709663573	73\\
-7.21863476037404	65\\
-6.93055242411235	75\\
-6.64247008785066	54\\
-6.35438775158898	37\\
-6.06630541532729	27\\
-5.7782230790656	20\\
-5.49014074280391	4\\
-5.20205840654222	11\\
-4.91397607028053	10\\
-4.62589373401884	6\\
-4.33781139775716	4\\
-4.04972906149547	0\\
-3.76164672523378	1\\
-3.47356438897209	0\\
-3.1854820527104	0\\
-2.89739971644871	0\\
-2.60931738018702	0\\
-2.32123504392533	1\\
-2.03315270766364	0\\
-1.74507037140196	0\\
-1.45698803514027	2\\
-1.16890569887858	1\\
-0.88082336261689	0\\
-0.592741026355201	1\\
-0.304658690093512	0\\
-0.0165763538318235	0\\
0.271505982429865	0\\
0.559588318691554	0\\
0.847670654953243	0\\
1.13575299121493	0\\
1.42383532747662	0\\
1.71191766373831	0\\
2	0\\
};
\addplot[ybar interval, fill=lightgray, draw=black, area legend] table[row 
sep=crcr] {%
x	y\\
-12.4041168130844	0\\
-12.1160344768228	0\\
-11.8279521405611	0\\
-11.5398698042994	0\\
-11.2517874680377	0\\
-10.963705131776	0\\
-10.6756227955143	0\\
-10.3875404592526	0\\
-10.0994581229909	0\\
-9.81137578672924	0\\
-9.52329345046755	0\\
-9.23521111420586	0\\
-8.94712877794418	0\\
-8.65904644168249	0\\
-8.3709641054208	0\\
-8.08288176915911	0\\
-7.79479943289742	0\\
-7.50671709663573	0\\
-7.21863476037404	0\\
-6.93055242411235	0\\
-6.64247008785066	0\\
-6.35438775158898	0\\
-6.06630541532729	0\\
-5.7782230790656	0\\
-5.49014074280391	0\\
-5.20205840654222	0\\
-4.91397607028053	0\\
-4.62589373401884	0\\
-4.33781139775716	0\\
-4.04972906149547	0\\
-3.76164672523378	0\\
-3.47356438897209	0\\
-3.1854820527104	0\\
-2.89739971644871	1\\
-2.60931738018702	1\\
-2.32123504392533	2\\
-2.03315270766364	3\\
-1.74507037140196	5\\
-1.45698803514027	19\\
-1.16890569887858	24\\
-0.88082336261689	70\\
-0.592741026355201	92\\
-0.304658690093512	166\\
-0.0165763538318235	193\\
0.271505982429865	139\\
0.559588318691554	97\\
0.847670654953243	65\\
1.13575299121493	34\\
1.42383532747662	19\\
1.71191766373831	70\\
2	70\\
};
\end{axis}
\end{tikzpicture}%
\\
    \caption{Histograms of the values of $\log(\alpha)$, in light gray, and 
$\log(|\alpha - \alpha^*|/\alpha^*)$, in dark gray, over 1000 
tests.}\label{fig:Ex1_2}
\end{center}
\end{figure}
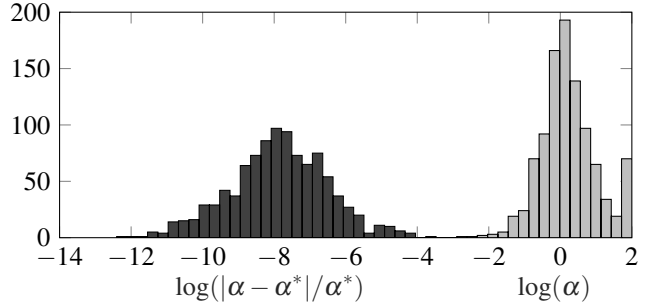

Finally, the Minkowski sum has been employed to compute $Q_N^x(\alpha \Omega, 
U, X)$ a posteriori, for evaluating its computational cost, but we could not go 
further than $n = 4$. 

\subsection{Example 2}

We apply now the proposed method to an high dimensional system, in particular 
with $n = 30$ and $m = 15$ with horizons $N = 5, 10$. To provide some hints on 
the conservatism of the control invariant obtained with respect to the maximal 
control invariant set, we build a system for which the latter can be computed. 
Indeed the classical algorithms for computing the maximal control invariant set 
are too computationally demanding to be applied to high dimensional systems in 
general. Then, a specific structure has to be imposed to the system dynamics 
for computing the maximal control invariant set to be compared with our results. 
In particular, we consider system (\ref{eq:system}) with 
\begin{equation}\label{eq:AyBy}
 A = P^{-1} \!\! \left[\begin{array}{cccc}
 A_1 & 0 & \ldots & 0\\
 0& A_2 & \ldots & 0\\
 \ldots & \ldots & \ldots & \ldots \\
 0 & 0 & \ldots & A_{15}\\
 \end{array}\right]\!\! P \!\!, \quad  
 B = P^{-1} \!\! \left[\begin{array}{ccccc}
 B_1 & 0 & \ldots & 0\\
 0& B_2 & \ldots & 0\\
 \ldots & \ldots & \ldots & \ldots \\
 0 & 0 & \ldots & B_{15}\\
 \end{array}\right]
\end{equation}
where $A_i \in \R^{2 \times 2}$ and $B_i \in \R^{2}$, for $i \in \N_{15}$, are 
matrices whose entries are randomly generated such that all $A_i$ have instable 
poles and the pairs $(A_i, B_i)$ are controllable and the maximal control 
invariant is obtained, as illustrated below, after 5 iterations at most. The 
latter requirement has been introduced for sets convergence reasons. The matrix 
$P \in \R^{30 \times 30}$ is a randomly generated nonsingular matrix. 
Figure~\ref{fig:Ex2_AB} provides a graphical representation of $A$ and $B$, for 
which the maximal values ($15.303$ for $A$ and $49.0516$ for $B$) are depicted 
in white, the minimal ones ($-13.4866$ for $A$ and $-60.4621$ for $B$) are 
drawn in black, the other values are proportional degree of gray. The matrices 
are not sparse, not a single null entry is present either in $A$ or $B$, and 
are available under request.  

\begin{figure}[H]
\begin{center}
    \setlength\fheight{3cm}
    \setlength\fwidth{5.25cm}
%
%
\begin{tikzpicture}

\begin{axis}[%
width=0.593\fwidth,
height=\fheight,
at={(0\fwidth,0\fheight)},
scale only axis,
point meta min=0,
point meta max=1,
axis on top,
xmin=0.5,
xmax=30.5,
tick align=outside,
y dir=reverse,
ymin=0.5,
ymax=30.5,
axis line style={draw=none},
ticks=none
]
\addplot [forget plot] graphics [xmin=0.5, xmax=30.5, ymin=0.5, ymax=30.5] 
{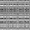};
\end{axis}
\end{tikzpicture}%
\hspace{1cm}
    \setlength\fheight{3cm}
    \setlength\fwidth{4.5cm}
%
%
\begin{tikzpicture}

\begin{axis}[%
width=0.374\fwidth,
height=\fheight,
at={(0\fwidth,0\fheight)},
scale only axis,
point meta min=0,
point meta max=1,
axis on top,
xmin=0.5,
xmax=15.5,
tick align=outside,
y dir=reverse,
ymin=0.5,
ymax=30.5,
axis line style={draw=none},
ticks=none
]
\addplot [forget plot] graphics [xmin=0.5, xmax=15.5, ymin=0.5, 
ymax=30.5]{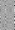};
\end{axis}
\end{tikzpicture}%
    \caption{Graphical representation of matrices $A$ and 
$B$.}\label{fig:Ex2_AB}
\end{center}
\end{figure}

Thus, the dynamics of system with state $y = P x$, is controllable and it is, 
in practice, composed by $15$ decoupled two-dimensional subsystems with one 
control input each. Hence, the maximal control invariant set in the space for 
the overall system in $y$, denoted $\Sigma$, is given by the Cartesian 
product of the maximal control invariant sets of the 15 subsystems. That is 
$\textstyle \Sigma =  \prod_{i = 1}^{15}\Sigma_i$ where 
$\Sigma_i$ are the maximal control invariant set in $10 \B^2$ for the 
$i$-th subsystem with input bound $10 \B^2$. Hence $\Sigma$ can be 
computed by computing $\Sigma_i$, being $(A_i, B_i)$ a two-dimensional 
controllable system, for all $i \in \N_{15}$. Therefore, $P^{-1} 
\Sigma \subseteq \R^{30}$ is the maximal control invariant set for the 
system (\ref{eq:system}) in $x$ with (\ref{eq:AyBy}). After computing $P^{-1} 
\Sigma$, the linear problem (\ref{eq:FinalCondLPbeta}) has been solved 
to obtain $\bar{\Omega}^{x}$ with $N = 5, 10$ and sets $\Omega = P^{-1} 
\B^{30}$, $U = 10 \B^{15}$ and $X = 10 P^{-1} \B^{30}$.

To quantify the difference between the maximal control invariant set $P^{-1} 
\Sigma$ and the set $\bar{\Omega}^{x}$, $100$ vectors $v \in \R^n$ 
are generated randomly. Then, (a lower approximation of) the maximal values of 
$r_{\Sigma}$ and $r_{\Omega}$ are computed such that $r_{\Sigma} v \in P^{-1} 
\Sigma$ and $r_{\Omega} v \in \bar{\Omega}^{x}$, through dichotomy 
method. In practice, we search for (approximations of) the intersections 
between the ray $v_r = \{r v \in \R^n: \ r \geq 0\}$ and the boundaries of the 
sets $P^{-1} \Sigma$ and $\bar{\Omega}^{x}$. The ratio between 
$r_{\Omega} / r_{\Sigma}$ is an indicator of the mismatch between the maximal 
control invariant set $P^{-1} \Sigma$ and $\bar{\Omega}^{x}$, the 
closer to one, the closer are the intersections between the ray $v_r$ and the 
two sets.

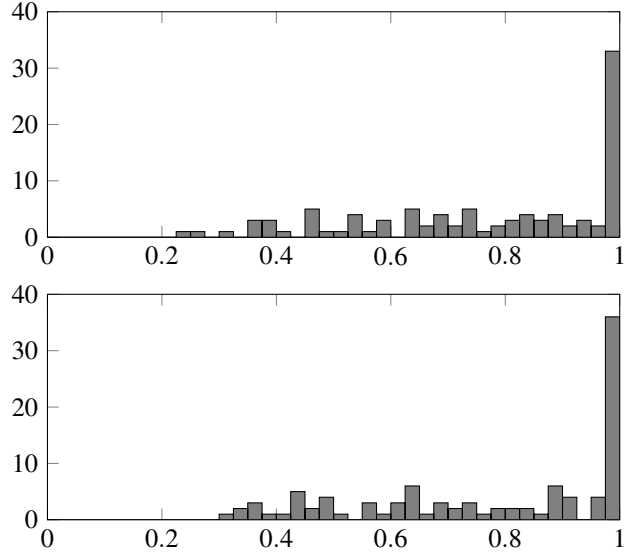
\begin{figure}
\begin{center}
    \setlength\fheight{3cm}
    \setlength\fwidth{8cm}
%
%
\begin{tikzpicture}

\begin{axis}[%
width=0.951\fwidth,
height=\fheight,
at={(0\fwidth,0\fheight)},
scale only axis,
xmin=0,
xmax=1,
ymin=0,
ymax=40,
axis background/.style={fill=white},
]
\addplot[ybar interval, fill=gray, draw=black, area legend] table[row sep=crcr] 
{%
x	y\\
0	0\\
0.025	0\\
0.05	0\\
0.075	0\\
0.1	0\\
0.125	0\\
0.15	0\\
0.175	0\\
0.2	0\\
0.225	1\\
0.25	1\\
0.275	0\\
0.3	1\\
0.325	0\\
0.35	3\\
0.375	3\\
0.4	1\\
0.425	0\\
0.45	5\\
0.475	1\\
0.5	1\\
0.525	4\\
0.55	1\\
0.575	3\\
0.6	0\\
0.625	5\\
0.65	2\\
0.675	4\\
0.7	2\\
0.725	5\\
0.75	1\\
0.775	2\\
0.8	3\\
0.825	4\\
0.85	3\\
0.875	4\\
0.9	2\\
0.925	3\\
0.95	2\\
0.975	33\\
1	33\\
};
\end{axis}
\end{tikzpicture}%
\\
    \setlength\fheight{3cm}
    \setlength\fwidth{8cm}
%
%
\begin{tikzpicture}

\begin{axis}[%
width=0.951\fwidth,
height=\fheight,
at={(0\fwidth,0\fheight)},
scale only axis,
xmin=0,
xmax=1,
ymin=0,
ymax=40,
axis background/.style={fill=white},
]
\addplot[ybar interval, fill=gray, draw=black, area legend] table[row sep=crcr] 
{%
x	y\\
0	0\\
0.025	0\\
0.05	0\\
0.075	0\\
0.1	0\\
0.125	0\\
0.15	0\\
0.175	0\\
0.2	0\\
0.225	0\\
0.25	0\\
0.275	0\\
0.3	1\\
0.325	2\\
0.35	3\\
0.375	1\\
0.4	1\\
0.425	5\\
0.45	2\\
0.475	4\\
0.5	1\\
0.525	0\\
0.55	3\\
0.575	1\\
0.6	3\\
0.625	6\\
0.65	1\\
0.675	3\\
0.7	2\\
0.725	3\\
0.75	1\\
0.775	2\\
0.8	2\\
0.825	2\\
0.85	1\\
0.875	6\\
0.9	4\\
0.925	0\\
0.95	4\\
0.975	36\\
1	36\\
};
\end{axis}
\end{tikzpicture}%
\\
    \caption{Histograms of the values $r_{\Omega} / r_{\Sigma}$ for $N = 5, 
10$ for a systems with $n = 30$ and $m = 15$.}\label{fig:Ex2_2}
\end{center}
\end{figure}

Figure~\ref{fig:Ex2_2} shows the histograms of the ratio $r_{\Omega} / 
r_{\Sigma}$ for $N = 5, 10$. As expected, the higher is the horizon $N$, 
the closer are the sets $\Sigma$ and $\bar{\Omega}^{x}$.

\section{Conclusions}

In this paper we addressed the problem of computing control invariant sets for 
linear systems with state and input polyhedral constraints. Invariance 
conditions are given, that are set inclusions involving the N-step sets, which 
are posed in form of LP optimization problems, instead of Minkowski sum of 
polyhedra. Then the procedures based on those conditions are applicable even 
for high dimensional systems. 

\bibliographystyle{plain}
\bibliography{SimplyInv_AUT}

\appendix

\section{Appendix}\label{app:A}    

\begin{lemma}\label{lem:lambdaK}
Given $K \subseteq \N_N$ and defined $\bar{K} = \N_n / K$ and 
\begin{equation}\label{eq:lambdaK}
 \Lambda(K) = \left\{\lambda \in \R^n: \ \lambda_k > 0 \ \forall k \in K, 
\quad \lambda_k = 0 \ \forall k \in \bar{K} \right\}
\end{equation}
one has 
\begin{equation}\label{eq:lambdaK2}
\begin{array}{l}
\displaystyle \{\lambda \in \R^n: \ \lambda_k \geq 0 \ \forall k \in \N_n, \ 
\mathbf{1}^T \lambda = 1\} \\ 
\hspace{1cm} = \bigcup_{K \subseteq \N_N} \{\lambda \in 
\Lambda(K): \mathbf{1}^T \lambda = 1\}
\end{array}
\end{equation}
\end{lemma}

\begin{proof}
Note that for every $\lambda \in \Lambda(K)$, $\lambda_k$ is strictly positive 
if and only if $k \in K$, i.e. $K$ denotes the set of indices such that 
$\lambda_k$ is not zero, in practice. For every $\lambda$ in the l.h.s. of 
(\ref{eq:lambdaK2}), there exists a $K$, that is the set of indices for which 
$\lambda_k > 0$, such that $\lambda \in \Lambda(K)$. Analogously, every 
$\lambda$ in the r.h.s. of (\ref{eq:lambdaK2}), also satisfies $\lambda \geq 0$ 
and then it is contained in the l.h.s. set. 
\end{proof}

\end{document}